\documentclass[Phys,submission]{SciPost_corr}
\usepackage{microtype}

\usepackage{multirow}
\usepackage{amssymb,amsthm,amsmath}
\usepackage{mathrsfs}
\usepackage{graphicx}
\usepackage{amsfonts}
\usepackage{color}
\usepackage{bm}
\usepackage{calc}
\usepackage{ifthen}
\usepackage[per-mode=symbol-or-fraction,add-decimal-zero = false,detect-all]{siunitx} 
\usepackage{tikz}
\usepackage[width=0.95\textwidth]{caption}
\newcommand{\maxim}[2]{#2}
\newcommand{\chris}[2]{#2}

\newcommand{\IsSubmission}{0}

\ifnum \IsSubmission=1
\else
\usepackage[fit]{truncate} 
\fancypagestyle{FirstPage}{\rhead{}} 
\fi

\protected\def\verythinspace{%
  \ifmmode
    \mskip0.3\thinmuskip
  \else
    \ifhmode
      \kern0.050004em
    \fi
  \fi
}

\DeclareMathOperator{\sign}{sign}

\theoremstyle{plain}
\newtheorem{theorem}{Theorem}[section]

\newtheorem{lemma}[theorem]{Lemma}
\newtheorem{proposition}[theorem]{Proposition}

\theoremstyle{definition}

\newtheorem{remark}[theorem]{Remark}

\begin{document}
\boldmath
\begin{center}{\Large \textbf{
Growth of the Wang-Casati-Prosen counter in an integrable billiard
}}\end{center}
\unboldmath
\begin{center}
Z. Hwang\textsuperscript{1},
C. A. Marx\textsuperscript{2},
J. Seaward\textsuperscript{3},
S. Jitomirskaya\textsuperscript{4}
M. Olshanii\textsuperscript{1*},
\end{center}
\begin{center}
{\bf 1} Department of Physics, University of Massachusetts Boston, Boston, MA 02125, USA
\\
{\bf 2} Department of Mathematics, Oberlin College, Oberlin, OH 44074, USA
\\
{\bf 3} Universit\'e Sorbonne Paris Nord, CNRS UMR 7538, Laboratoire de physique des lasers, F-93430, Villetaneuse, France
\\
{\bf 4} Department of Mathematics, University of California, Irvine, CA 92697-3875, USA
\vskip 0.25\baselineskip
{\scriptsize *} maxim.olchanyi@umb.edu
\end{center}
%
%
%
%
%
\boldmath
\section*{Abstract}
{\bf
This work is motivated by an article by Wang, Casati, and Prosen [Phys.~Rev.~E vol.~89, 042918 (2014)] devoted to a study of ergodicity in two-dimensional irrational right-triangular billiards. Numerical results presented there suggest that these billiards are generally not ergodic. However, they become ergodic when the billiard angle is equal to $\pi/2$ times a Liouvillian irrational, a Liouvillian irrational, a class of irrational numbers which are well approximated by rationals. 

In particular, Wang et al.\ study a special integer counter that reflects the irrational contribution to the velocity orientation; they conjecture  that this counter is localized in the generic case, but grows in the  Liouvillian  case. We propose a generalization of the Wang-Casati-Prosen counter: this generalization allows to include rational billiards into consideration. We show that in the case of a  $\ang{45} \!\! : \!\ang{45} \!\! : \! \ang{90}$ billiard, the counter grows indefinitely, consistent with the Liouvillian scenario suggested by Wang et al. 

}
\unboldmath
\vspace{10pt}
\noindent\rule{\textwidth}{1pt}
\tableofcontents
\noindent\rule{\textwidth}{1pt}
\vspace{10pt}

\ifnum \IsSubmission=1
\else
\thispagestyle{FirstPage} 
\fi


\section{Introduction}
\label{sec:introduction}

In Ref.~\cite{wang2014_042918} Wang, Casati, and Prosen studied ergodicity in two-dimensional irrational right-triangular billiards. The numerical results presented there suggest that while these billiards are not ergodic in general, they become ergodic when the billiard angle is equal to $\pi/2$ times a Liouvillian irrational. The latter is a class of numbers with properties lying in between irrational and rational. 

Authors present an elegant semi-empirical construction that sheds light to the mechanisms behind the ergodicity breaking. Using numerical evidence, they conjecture that for a given rational approximant of the billiard angle (in units of $\pi/2$) ergodicity requires an exponentially long time to establish, while the validity of the approximant expires in a linear time, both in terms of approximant's denominator. They state  a very modest necessary condition for the ergodicity of the original, irrational billiard: for each rational approximant, ergodicity must be reached before approxomant's \emph{successor} becomes invalid. A generic irrational number for the billiard angle won't pass this test. However, Liouville numbers---numbers for which the error of a rational approximation decreases faster than any negative power of its denominator---produce billiards that satisfy the above necessary condition. If the propagation time is bounded from above, Liouville builliards are indistinct from the rational ones: expectedly, the rational billiards  themselves  also pass the test.      

As a quantitative measure of ergodicity, Wang et al. introduce a special integer counter (that we will call, in this text, a Wang-Casati-Prosen counter) that reflects the irrational contribution to the velocity orientation. Ergodicity requires that this counter grows indefinitely with time. Authors provide a numerical evidence that the counter is localized in the generic irrational billiards. They further conjecture that the counter grows in the Liouvillian case. They base this conjecture on the fact that Liouville billiards satisfy the thecessary condition for the ergodicity described in the previous paragraph. This former conjecture is also consistent with the rigorous results presented in \cite{vorobets1996_756}, for a subset of Liouville numbers.    

Motivated by Wang et al., here  we study an extencion of the Wang-Casati-Prosen counter for a $\ang{45} \!\! : \!\ang{45} \!\! : \! \ang{90}$ billiards---a rational billiard featuring only eight velocity orientations for any generic trajectory---and find, in accordance with the scenario suggested in \cite{wang2014_042918} for the Liouville numbers, that this counter grows indefinitely. 
We prove the absence of localization by identifying a subsequence along which the Wang-Casati-Prosen  counter 
provably grows, albeit in a logarithmic fashion.  

At a technical level, analytical results we obtained became possible thanks to the solvability of the $\ang{45} \!\! : \!\ang{45} \!\! : \! \ang{90}$ billiards using the method of images.   

\section{Motivation: Wang-Casati-Prosen counter in right-triangular billiards. First appearance of the rational billiards}
\label{sec:motivation}
In \cite{wang2014_042918}, authors consider a single two-dimensional point-like particle moving in a right-triangular billiards (see Fig.~\ref{f:WCP}) of an acute angle 
\maxim{$\tilde{\alpha}$ ($\alpha$ in the original)}{$\tilde{\alpha}/2 = ((\sqrt{5}-1)/4) \pi$ ($\alpha/2$ in the original)}. 
\begin{figure}[h]
\begin{center}
\includegraphics[width=0.5\textwidth]{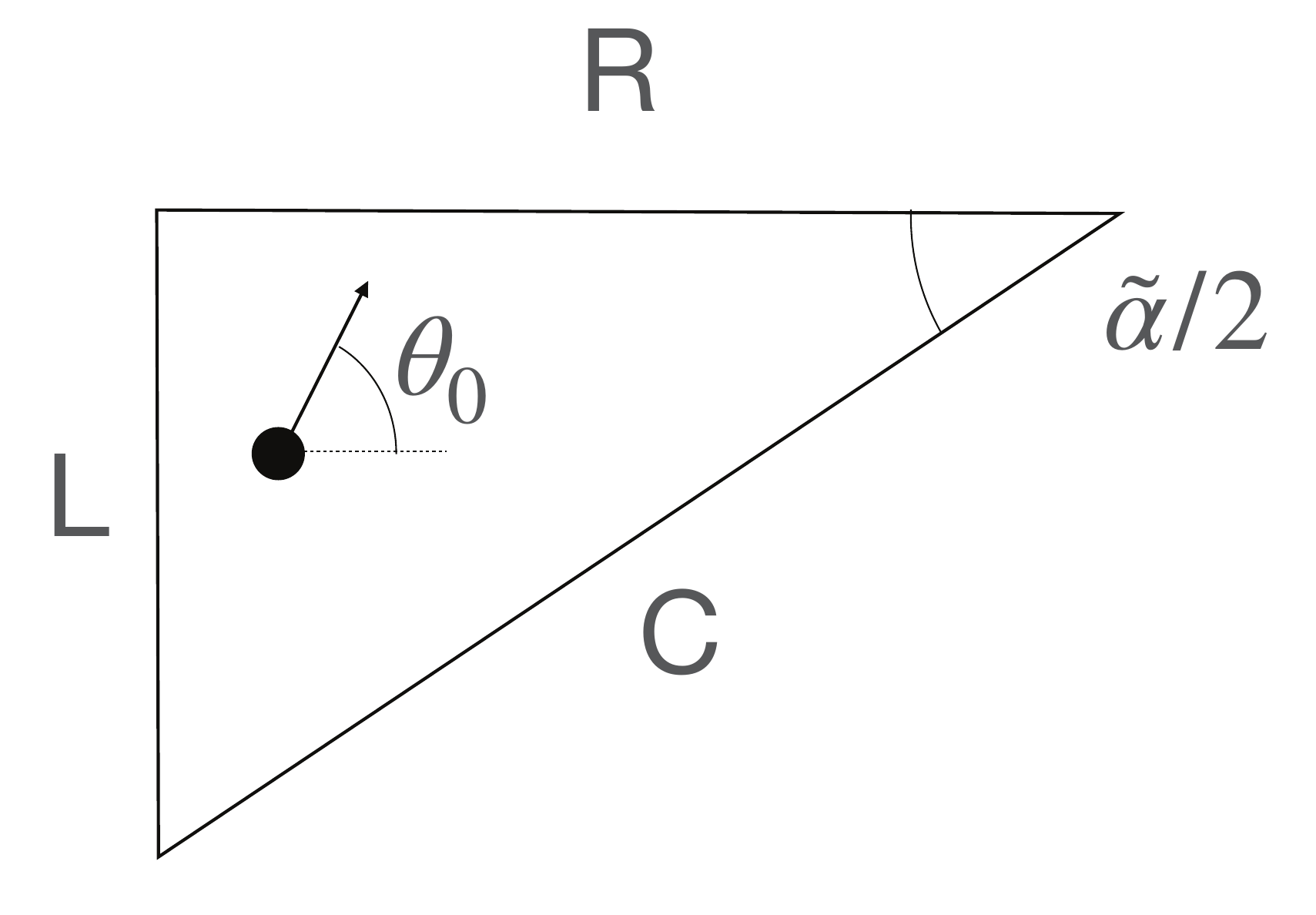}
\end{center}
\caption
{
The right-triangular billiards considered in \cite{wang2014_042918}. $\theta_{0}$ gives the orientation of the initial velocity. $\tilde{\alpha}/2$
is one of the angles. See
text for the rationale behind the names of the sides of the triangle.  
}
\label{f:WCP}
\end{figure}
The physical model behind the right-triangular billiards consists of two one-dimensional hard-core point-like particles between two walls. The 
map is described in \cite{wang2014_042918} and elsewhere. According to the map, the vertical cathetus, the horizontal cathetus, 
and the hypothenuse correspond respectively to the left particle colliding with a wall ($L$), the right particle colliding with a wall ($R$), and a particle-particle 
collision ($C$); hence the naming of the sides in Fig.~\ref{f:WCP} and throughout the text. 

A particle's trajectory starts at some point inside the billiards, with a velocity vector at an angle $\theta_{0}$ to the horizontal axis. The authors of \cite{wang2014_042918}
observe that at any instant of time, the angle $\theta$ between the velocity and the horizontal has the form
\begin{align*}
\theta = K \tilde{\alpha} + 
\left\{ 
\begin{array}{l} 
+\theta_{0}
\\
\text{or}
\\
-\theta_{0}
\\
\text{or}
\\
+(\pi-\theta_{0})
\\
\text{or}
\\
-(\pi-\theta_{0})
\end{array} \right 
\}
\,\,,
\end{align*}
with the \emph{counter} $K$ being an integer. The paper studies the growth of the magnitude of $K$ over time. The numerical evidence presented 
\maxim{}{in \cite{wang2014_042918}} suggests that $K$ is localized within 
a $\sim\pm 20$ range around the origin, suggesting absence of ergodicity. 
\maxim{}{The authors suggest that this phenomenon is general for generic $\tilde{\alpha}$, while conjecturing an absence of localization for
$\tilde{\alpha} = \text{Liouvillian\_number} \times \pi$.} 
\maxim{In our paper we study an analogue of the counter $K$ for the case $\tilde{\alpha} = 90^{\circ}$, 
and find no localization.}{In our paper, we study an analogue of the counter $K$ for the \emph{rational} case of $\tilde{\alpha} = \pi/2$. 
We find no localization, in accordance with the expectations of \cite{wang2014_042918}.}

To proceed, observe that the velocity orientation, and, thus, the counter value does not change between the particle-wall collisions. Hence,  the counter $K$ is a function 
of the number $i$ of particle-wall collisions prior. As such, the temporal index $i$ labels the time intervals between two successive partcle-wall collisions. Accordingly, from now on, we will denote $K$ as $K_{i}$. It is easy to show (see  \cite{wang2014_042918}) that the rule 
for updating the counter $K_{i}$ is as follows:
\begin{align}
\begin{split}
&
K_{i+1} = 
\left\{
\begin{array}{ll}
-K_{i} + 1\,, & \begin{minipage}[t]{0.4\textwidth}
                 if the event that separates the $i+1$'st and $i$'th time interval is $C$\\
                 \mbox{}
                \end{minipage}
\\
-K_{i}\,, & \begin{minipage}[t]{0.4\textwidth}
if the event that separates the $i+1$'st and $i$'th time interval is $L$ or $R$ 
\end{minipage}
\end{array}
\right\}
\\
&
K_{0} = 0\,.
\end{split}
\label{K}
\end{align}

This is a \emph{crucial moment}. Let's try to assign a rational value to $ \tilde{\alpha}/\pi$. The meaning of the 
counter $K$ as the value of the $\tilde{\alpha}/\pi$ contribution to the velocity orientation will be completely lost: there will be an infinite 
multiplicity of values of $K$ leading to the same velocity orientation. However, the dynamics of the counter $K$ \emph{itself} remains nontrivial, constituting 
a viable object of study. More interestingly, it remains nontrivial (i.e.\ no closed form formula exists for the counter $K$ as a function of $i$, for a given initial condition) even the for the two \emph{integrable} 
right-triangular billiards: $\ang{30} \!\! : \!\ang{60} \!\! : \! \ang{90}$ an $\ang{45} \!\! : \!\ang{45} \!\! : \! \ang{90}$ respectively. The latter billiard is the one our article focusses on.

\section{Irrational rotations}
\label{sec:object_of_study}

Given an $\alpha: 0\le\alpha<1$, we consider the $\alpha$-rotational trajectory emerging for $x_0 \in [0,1)$,
\begin{align} \label{x}
x^{(\alpha)}_{j+1} &= ( x^{(\alpha)}_{j} + \alpha )\mod 1 ~\mbox{, } j \in \mathbb{N}_0 ~\mbox{.} \\
x^{(\alpha)}_{0} &= x_0 ~\mbox{,}
\end{align}
where $\mathbb{N}_0:= \mathbb{N} \cup \{0\}$.

For a fixed $\beta: 0\le \beta < 1$, introduce an ``observable'' $f^{(\beta)}: [0,1) \to   \maxim{[0,1)}{\{-1,\,+1\}}$, 
\begin{align}
f^{(\beta)}(x) =
\left\{
\begin{array}{ccl}
+1\;, & \text{for} & x \in \mathcal{I}_{\text{I}} := [0, \beta) ~\mbox{,}
\\
-1\;, & \text{for} & x \in \mathcal{I}_{\text{II}} := [\beta , 1) ~\mbox{.}
\end{array}
\right.
\label{f_function}
\,\,
\end{align}
and a corresponding sequence
\begin{align} \label{f}
f^{(\alpha,\,\beta)}_{j} :=  f^{(\beta)}(x^{(\alpha)}_{j})  ~\mbox{, } j \in \mathbb{N}_0 ~\mbox{.}
\end{align}
Finally, for each $j \in \mathbb{N}_0$, consider an ``increment''
\begin{align} \label{epsilon}
\epsilon^{(\alpha,\,\beta)}_{j} := \chi(f^{(\alpha,\,\beta)}_{j}) \prod_{j'=0}^{j} f^{(\alpha,\,\beta)}_{j'} ~\mbox{,}
\end{align}
with
\begin{align}
\chi(f) :=
\left\{
\begin{array}{ccc}
2\;, & \text{for} & f = +1
\\
1\;, & \text{for} & f = -1
\end{array}
\right.
\label{eta_function}
\,\,.
\end{align}
Our primary object of interest is the following ``counter'':
\begin{align} \label{S}
S^{(\alpha,\,\beta)}_{j} = \sum_{j'=0}^{j-1} \epsilon^{(\alpha,\,\beta)}_{j'} ~\mbox{, } j \in \mathbb{N}_0 ~\mbox{.}
\end{align}
A similar object was considered in \cite{veech1968}.
There, $\mathcal{I}_{\text{II}}$ was any connected nonempty subset of $[0,1)$ and $\chi(f) = 1$. 

\chris{%
}
{
We recognize that there is a wide range of papers in the mathematics literature, written from a more general point of view, which are dedicated to studying the ergodic properties of billiards and cocycle maps of the type considered here. As mentioned earlier, our model can be considered a part of a certain family of dynamical systems which were first investigated in 1968 by W.A. Veech in \cite{veech1968}. Without aiming for a complete list, we refer the reader to e.g. \cite{Conze_2009,Troubetzkoy_2004} and \cite{Squillace_2017, Cheung_Hubert_Masur_2011, DelecroixHubertLelievre_2014} for a few more recent general accounts of recurrence and ergodicity of cocycles over a rotation in relation to the model considered here.

Rather than aiming for generality, the purpose of this work is to construct {\emph{explicit}} analytical examples for a concrete physical model of interest, for which the numerical studies in \cite{wang2014_042918} indicated an interesting behavior of the counter depending on Diophantine properties. 
}

Specifically, our main result, Theorem \ref{thm_main}, considers the ``diagonal'' case
\begin{align*}
\alpha = \beta =: \sigma ~\mbox{,}
\end{align*}
\chris{%
for which we prove:%
}
{%
for which, motivated by the results in \cite{wang2014_042918}, we provide analytical evidence for a subtle dependence of the counter on the Diophantine properties of the parameter $\sigma$. Specifically, we explicitly describe a set of irrational values of $\sigma$ for which we can verify indefinite growth of the counter for a full-measure set of initial conditions:%
}
\begin{theorem} \label{thm_main_main}
There exists a non-empty set of irrationals $\Sigma \subseteq (0,1)$ \chris{%
such that for each %
} 
{%
so that for each $\sigma \in \Sigma$, the following holds: %
}
There is an associated full measure set of initial conditions $\Omega(\sigma)$ such that for all $x_0 \in \Omega(\sigma)$, one has
\begin{equation} \label{eq_thm_main_main}
\limsup_{j \to \infty} \left|S_{j}\right| = \infty ~\mbox{.}
\end{equation}
\end{theorem}
\chris{%
}
{%
\begin{remark}
The set $\Sigma$ in Theorem \ref{thm_main_main} is described explicitly in terms of continued fraction expansion of its elements, see Theorem \ref{thm_main} of Section \ref{sec:main}. The description of the full measure set of initial conditions $\Omega(\sigma)$ is given in Remark \ref{remark_fullmeasinit} of Section \ref{sec:main}. It remains an interesting open question whether the conditions on $\sigma$ and $x_0$ are indeed necessary and whether the growth along a subsequence could be strengthened to growth along the entire sequence, which would replace the $\limsup$ in (\ref{eq_thm_main_main}) by a limit; see also our remarks at the end of Section \ref{sec:main}.
\end{remark}

}

\section{Connection between the irrational rotations and the Wang-Casati-Prosen counter in a  $\ang{45} \!\! : \!\ang{45} \!\! : \! \ang{90}$ billiard}
\label{sec:method}
Let us return back to the end of the Section~\ref{sec:motivation}.
Without loss of generality, assume that $i=0$ labels an interval between an $L$ and an $R$ event, the former preceded by $C$ (see Fig.~\ref{f:WCP}). Name the two events preceding the interval $i=0$ as $C^{\star}$ and $L^{*}$, 
for future reference. Introduce a counter 
\begin{align*}
P_{i} := (-1)^{i+1} K_{i} \,,
\,\,
\end{align*}
a temporal index 
\begin{align*}
\bar{m}(i) := \text{\# $C$-events between the $C^{\star}$-event and the $i$'th interval, excluding $C^{\star}$}
\,\,,
\end{align*}
and another temporal index 
\begin{align*}
i'(m) := \inf\left\{ i : \bar{m}(i) = m  \right\}
\,\,.
\end{align*}
Observe that the counter $P_{i}$ is a function of $\bar{m}(i)$ alone as it does not change after $L$ and $R$ events. Accordingly, introduce a counter 
\begin{align*}
Q_{m} := P_{i'(m)} 
\,\,.
\end{align*}
Recall that, by construction, 
\begin{align*}
|Q_{m}| = |K_{i'(m)}| 
\,\,,
\end{align*}
so that if $Q_{m}$ is found to be unbounded, then $K_{i}$ will be unbounded as well. 

Note also that $m$ labels a temporal interval between two successive $C$ events. A crucial observation is that the dynamics of the counter $Q_{m}$ is governed by
\begin{align}
\begin{split}
&
Q_{m+1} = Q_{m} + \varepsilon_{m} 
\\
&
Q_{0} = 0
\end{split}
\,\,,
\label{Q}
\end{align}
with
\begin{align}
\varepsilon_{m} := (-1)^{M}
\,,
\label{varepsilon}
\end{align}
where
\begin{align*}
M := \begin{minipage}[t]{0.6\textwidth}
      \centering
      \# $CLRC $- or $CRLC$-intervals between the $C^{\star}$-event, inclusive, and the $m$'th interval, inclusive.
     \end{minipage}
\end{align*}

So far, our discussion concerned a generic right triangle of Fig.~\ref{f:WCP}. From now on, let us assume a $45^{\circ}\!\!:\!45^{\circ}\!\!:\!90^{\circ}$ billiards, i.e.~that
\begin{align}
\tilde{\alpha} = \frac{\pi}{2}
\,\,.
\label{alpha_tilde_fixed}
\end{align}
Also assume, without loss of generality, that 
\begin{align}
\frac{\pi}{4} \le \theta_{0} < \frac{\pi}{2}
\,\,.
\label{theta0_fixed}
\end{align}

Thanks to the method of images---described in the caption of Fig.~\ref{f:ZAI}---particle's trajectory becomes fully predictable, even for long propagation times, with no sensitivity 
to the initial conditions. The evolution of the counter $Q$ (see \eqref{Q}-\eqref{varepsilon}) remains sensitive to the initial conditions. During a sequence of $CLRC$ and $CRLC$ fragments, 
the counter $Q$ lingers around a particular value, with no substantial evolution. Such sequences are interrupted by an occasional (isolated, given \eqref{alpha_tilde_fixed}-\eqref{theta0_fixed}) $CRCLC$ fragment that changes the value of the counter by $\pm 2$. Whether the change keeps the sign of the previous $\pm 2$ jump or flips it  depends on the parity of the number of the $CLRC$ and $CRLC$  fragment in between. This parity, 
in turn, can be altered by a small change in the initial position, leading to large deviations at long propagation times.  
\begin{figure}[h]
\begin{center}
\includegraphics[width=0.8\textwidth]{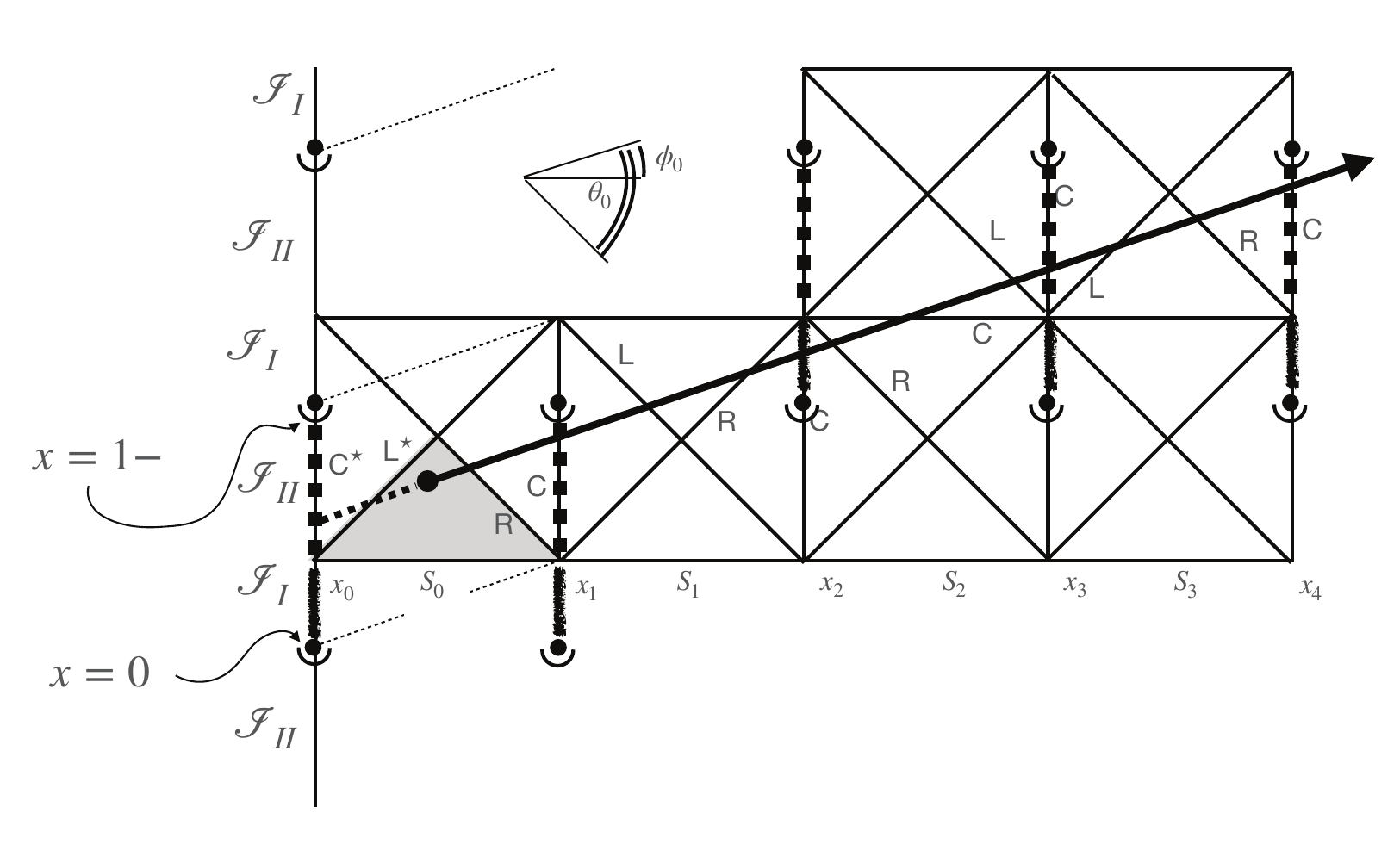}
\end{center}
\caption
{
The relationship between the dynamics in a $45^{\circ}\!\!:\!45^{\circ}\!\!:\!90^{\circ}$ billiards and 
irrational rotations. The grey triangle corresponds to the billiards in question. The white triangles are the images of the original billiards 
obtained via mirror reflections with respect to the wall that the particle hits. The subsequent fragment of the particle's trajectory is reflected as well; as a result, the trajectory becomes a straight line traveling through 
a plane tiled by the copies of the original billiards. 
The large black circle corresponds to the initial position of the particle. The vertical dotted and ``hand-drawn'' lines correspond to two intervals in the related irrational rotation model. $x_{j}$
is the phase space variable in the irrational rotations. See the text for the definition of the counter $S_{j}$. The time $j$ counts the vertical $C$-lines crossed by the trajectory.  
}
\label{f:ZAI}
\end{figure}

Let us introduce another
temporal index,
\begin{align*}
&
\bar{j}(m) :=
\begin{minipage}[t]{0.6\textwidth}
\centering
 \# $CLRC$, $CRLC$, and $CRCLC$ fragments between the $C^{\star}$-event and the $m$'th interval.
\end{minipage}
\end{align*}
Note that $\bar{j}(m)$ counts the number of ``vertical'' $C$-events between $C^{\star}$, exclusive, and the $m$'th interval (see Fig.~\ref{f:ZAI}). 

Accordingly, introduce
\begin{align*}
m'(j) := \inf\left\{ m : \bar{j}(m) = j  \right\}
\,\,.
\end{align*}
Observe now that the counter $Q_{m}$ is a function of $\bar{j}(m)$ alone. Finally, introduce the counter 
\begin{align*}
S_{j} := Q_{m'(j)} 
\,\,.
\end{align*}
Again, by construction, 
\begin{align*}
|S_{j}| = |K_{i'(m'(j))}| 
\,\,,
\end{align*}
and a growth of the artificially constructed counter $S_{j}$ would indicate a growth of the physical counter of interest, $K_{i}$. 

Recall that $j$ labels the temporal intervals between successive ``vertical'' $C$ events. Let it also label the left ``vertical'' $C$ event for a given $j$'th interval. 
Without loss of generality assume that the hypothenuse $C$ has unit length:
\begin{align*}
\text{length}(C) = 1
\,\,.
\end{align*}
Introduce 
\begin{align*}
x_{j} := \text{position of the \maxim{}{``vertical''} particle-wall \maxim{$C$}{} collision}
\,\,,
\end{align*}
in \maxim{turns}{terms} of the tiling depicted in Fig.~\ref{f:ZAI}. 
\maxim{}{Notice that---for notational convenience---the origin $x = 0$ is a distance $\tan(\phi_{0})$ away from the top corner.}
It is easy to see that $x_{j}$ undergoes a sequence of irrational rotations, with a shift 
\begin{align*}
\sigma = \tan(\phi_{0})
\,\,,
\end{align*}
where
\begin{align*}
&
\phi_{0} := \theta_{0} - \frac{\pi}{4}
\\
&
0 \le \phi_{0} < \frac{\pi}{4} 
\,\,.
\end{align*}

Now, divide each of the ``vertical'' $C$-walls onto two areas
\begin{align*}
&
\mathcal{I}_{\text{I}} := [0,\,\sigma[
\\
&
\hspace{1em}\text{and}
\\
&
\mathcal{I}_{\text{II}} := [\sigma,\,1[
\,\,.
\end{align*}
Observe that if $x_{j}$ gets to the $\mathcal{I}_{\text{I}}$ area, a $CRCLC$ fragment will follow and the counter $S$ will change by $\pm 2$, with \emph{the same} sign as the preceding increment. 
(Recall that \eqref{alpha_tilde_fixed}-\eqref{theta0_fixed} dictate that $CLCRC$ fragments are impossible, and that each  $CRCLC$ fragment is isolated, i.e.\ it is surrounded 
by either $CLC$ or $CRC$ fragments.) Likewise, when $x_{j}$ is in  $\mathcal{I}_{\text{II}}$, a $CLC$ or $CRC$ fragment follows, and the counter $S$ changes by $\pm 1$, 
\emph{reversing} the sign of the previous change. All in all, it is easy to show that the counter $S_{j}$ follows the dynamics described by the equations \eqref{x}, \eqref{f_function}, \eqref{f}, \eqref{epsilon}, \eqref{eta_function}, and \eqref{S}.

\section{The method of study: auxiliary rational rotations}
\label{sec:method}
Let us return back to the end of the Section~\ref{sec:object_of_study}.
For a given irrational $\sigma: 0<\sigma<1$, consider its continued fraction expansion
\begin{align*}
\sigma = [0;\,a_{1},\,a_{2},\,\ldots] ~\mbox{,}
\end{align*}
and the corresponding convergents (rational approximants),
\begin{align*}
\sigma_{n} := [0;\,a_{1},\,a_{2},\,\ldots,\,a_{n}] = \frac{p_{n}}{q_{n}} ~\mbox{, } n \in \mathbb{N} ~\mbox{,}
\end{align*}
where $p_{n}$ and $q_{n}$ are mutually prime.

For $n  \in \mathbb{N}$, consider an auxiliary trajectory $(x^{(\sigma_{n})}_{j})_{j \in\mathbb{N}_0}$. Observe that this trajectory is $q_n$-periodic, i.e.
\begin{align} \label{x_periodicity}
x^{(\sigma_{n})}_{j+q_{n}} = x^{(\sigma_{n})}_{j} ~\mbox{, for all } j \in \mathbb{N}_0 ~\mbox{,}
\end{align}
whence so is the sequence of observables $(f^{(\sigma_{n},\,\sigma_{n})}_{j})_{j \in \mathbb{N}_0}$. The definition in (\ref{epsilon}) thus implies that $(\epsilon_{j}^{(\sigma_{n},\sigma_{n})})_{j \in \mathbb{N}_0}$ is either periodic or anti-periodic with period
$q_{n}$,
\begin{align*}
\epsilon^{(\sigma_{n},\,\sigma_{n})}_{j+q_{n}} = 
\eta^{(\sigma_{n},\,\sigma_{n})}_{n}
\epsilon^{(\sigma_{n},\,\sigma_{n})}_{j}
\,\,,
\end{align*}
where 
\begin{align*}
\eta^{(\alpha,\,\beta)}_{n} := \sign(\epsilon^{(\alpha,\,\beta)}_{q_{n}-1}) = \prod_{j=0}^{q_{n}-1} f^{(\alpha,\,\beta)}_{j}
\,\,.
\end{align*}

As a result, the sequence of corresponding counters either grows indefinitely, period after period, or remains trapped around zero:
\newlength{\alength}
\setlength{\alength}{\widthof{even}}
\begin{align*}
S^{(\sigma_{n},\,\sigma_{n})}_{N q_{n}} = 
\left\{
\begin{array}{ccc}
\hspace{5.2em} N\hspace{4.55em}\;, & \text{for} & \eta^{(\sigma_{n},\,\sigma_{n})}_{n} = +1
\\
\left\{
\begin{array}{ccl}
0\;, & \text{for} & N=\text{even}
\\
1\;, & \text{for} & N=\text{odd}
\end{array}
\right\}\;,
& \text{for} & \eta^{(\sigma_{n},\,\sigma_{n})}_{n} = -1 
\end{array}
\right\}
\cdot S^{(\sigma_{n},\,\sigma_{n})}_{q_{n}} 
\,\,,
\end{align*}
for any $N\in \mathbb{N}$. In particular, if $ \eta^{(\sigma_{n},\,\sigma_{n})}_{n} = +1$ and $S^{(\sigma_{n},\,\sigma_{n})}_{q_{n}} \neq 0$, the counter 
$S^{(\sigma_{n},\,\sigma_{n})}_{j}$ is obviously unbounded, since it is unbounded on the subsequence with $j = N q_{n}$, where it grows as
\begin{align}
S^{(\sigma_{n},\,\sigma_{n})}_{N q_{n}} \quad  \stackrel{\eta^{(\sigma_{n},\,\sigma_{n})}_{n} = +1}{=} \quad N S^{(\sigma_{n},\,\sigma_{n})}_{q_{n}}
\,\,.
\label{counter_growth_01}
\end{align}

This observation is the cornerstone of the proof of unboundedness of $S^{(\sigma,\,\sigma)}_{j}$ for appropriately chosen irrational rotations $\sigma$ and initial conditions $x_0$, which will be constructed in Section~\ref{sec:main} (Theorem \ref{thm_main}). 

\section{Growth for (rational) auxiliary rotations}
Fix $n  \in \mathbb{N}$ and consider the $q_n$-periodic auxiliary trajectory $(x^{(\sigma_{n})}_{j})_{j \in\mathbb{N}_0}$. The goal of this section is to explore the structure of the trajectory to quantify the growth of the counters $S^{(\sigma,\,\sigma)}_{j}$ over one period $0 \leq j \leq q_n-1$. Specifically, we will establish the following:
\begin{proposition} \label{lem:Sge2}
Suppose that {\emph{both}} $p_n$ and $q_n$ are odd. Then, one has the lower bound
\begin{align} \label{Sge2}
\left|S^{(\sigma_{n},\,\sigma_{n})}_{q_{n}}\right|\ge 2 ~\mbox{.}
\end{align}
\end{proposition}

To prove Proposition \ref{lem:Sge2}, we start by recalling some basic facts about rational rotations. First observe that since $\mathrm{gcd}(p_n,q_n) = 1$, the finite trajectory $\{x_j^{(\sigma_n)} ~,~ 0 \leq j \leq q_n-1\}$ consists of $q_n$ {\em{distinct}} points, each of which is visited only once. Moreover, identifying $\mathbb{R} / \mathbb{Z}$ with the unit circle $S^1 \subseteq \mathbb{C}$ via the bijection $t \mapsto \mathrm{e}^{2 \pi i t}$ and noticing that
\begin{equation}
\left( \mathrm{e}^{2 \pi i (x_0 - x_j^{(\sigma_n)}) } \right)^{q_n} = 1 
\end{equation}
shows that the points $\mathrm{e}^{2 \pi i (x_0 - x_j^{(\sigma_n)})}$ are merely permutations of the $q_n$-th roots of unity, i.e.
\begin{equation}
\{x_j^{(\sigma_n)} ~\mbox{, } 0 \leq j \leq q_n -1 \} = \{ x_0 + \frac{k}{q_n} ~\mbox{, } 0 \leq k \leq q_n - 1 \} ~\mbox{.}
\end{equation}

In particular, writing
\begin{equation} \label{k_0n}
x_0 = \left( x_0 \mod \frac{1}{q_n} \right) + \frac{\lfloor x_0 q_n \rfloor}{q_n} =: x_{0,n} + k_{0,n} \frac{1}{q_n} ~\mbox{,}
\end{equation}
we may represent the elements of the finite trajectory $\{x_j^{(\sigma_n)} ~,~ 0 \leq j \leq q_n-1\}$ in the form
\begin{align} \label{eq_represent_auxiltraj}
x_j^{(\sigma_n)} & = x_{0,n} + \frac{k_{j,n}}{q_{n}} ~\mbox{, } 0 \leq j \leq q_n -1 ~\mbox{,} 
\end{align}
with
\begin{equation} \label{k_j_n}
k_{j,n} := ( k_{0,n} + j p_{n} ) \mod q_{n} ~\mbox{.}
\end{equation}
The main merit of the representation of the auxiliary trajectory in (\ref{k_0n})-(\ref{k_j_n}) is that it allows to keep track of the value of the observable since
\begin{align} \label{eq_observable_k_j_n}
f^{(\sigma_{n},\,\sigma_{n})}_{j} = 
\left\{
\begin{array}{ccl}
+1\;, & \text{for} & k_{j,n} = 0,\,1,\,\ldots,\, p_{n}-1
\\
-1\;, & \text{for} & k_{j,n} = p_{n},\,\ldots,\,q_{n}-1
\end{array}
\right.
\,\,.
\end{align}

In particular, (\ref{eq_observable_k_j_n}) immediately yields
\begin{lemma} \label{lem:etaEqP1}
\begin{align*}
\eta^{(\sigma_{n},\,\sigma_{n})}_{n} = +1 \text{\rm\,\,  , for\,\,  } q_{n}-p_{n} = \text{\rm even} ~\mbox{.}
\end{align*}
\end{lemma}
\begin{proof}
For $0 \leq j \leq q_n-1$, $k_{j,n}$ will visit each of the $q_{n}$ points $0,\,1,\,\ldots,\,q_{n}-1$ precisely once. Therefore, (\ref{eq_observable_k_j_n}) implies that
\begin{align*}
\eta^{(\sigma_{n},\,\sigma_{n})}_{n} = \prod_{j=0}^{q_{n}-1} f^{(\sigma_{n},\,\sigma_{n})}_{j}
\end{align*}
is a product of $p_{n}$ factors $+1$ and $q_{n}-p_{n}$ factors $-1$. Thus, if $q_{n}-p_{n}$ is even, we conclude that $\eta^{(\sigma_{n},\,\sigma_{n})}_{n} = +1$.
\end{proof}

We are now ready to prove the main result of this section:
\begin{proof}[Proof of Proposition \ref{lem:Sge2}]
According to (\ref{eq_observable_k_j_n}), we have
\begin{align} \label{eq_epsilon_keylemma}
\epsilon^{(\sigma_{n},\,\sigma_{n})}_{j} = 
\left\{
\begin{array}{ccl}
+2\;, & \text{for} & k_{j,n} = 0,\,1\,\ldots,\, p_{n}-1
\\
-1\;, & \text{for} & k_{j,n} = p_{n},\,p_{n},\,\ldots,\,q_{n}-1
\end{array}
\right\}
\,
\cdot \sign(\epsilon^{(\sigma_{n},\,\sigma_{n})}_{j-1})
\,\,.
\end{align}

Using the fact that, for $0 \leq j \leq q_n-1$, $k_{j,n}$ will visit each of the $q_{n}$ points $0,\,1,\,\ldots,\,q_{n}-1$ precisely once, we may consider the subsequence $j_l$ consisting of all instances $j$ for which $f^{(\sigma_{n},\,\sigma_{n})}_{j} = -1$; in particular, along this subsequence one has 
\begin{align*}
\epsilon^{(\sigma_{n},\,\sigma_{n})}_{j_{l}} = -\epsilon^{(\sigma_{n},\,\sigma_{n})}_{j_{l-1}} 
\,\,.
\end{align*}

The sequence $\epsilon^{(\sigma_{n},\,\sigma_{n})}_{j_{l}} $ has an even number of terms (i.e.~$q_{n}-p_{n}$), and thus does not contribute to the sum $S^{(\sigma_{n},\,\sigma_{n})}_{q_{n}} = \sum_{j'=0}^{q_{n}-1} \epsilon^{(\sigma_{n},\,\sigma_{n})}_{j}$. By (\ref{eq_epsilon_keylemma}), the remaining summands are all $\pm 2$, and there is an odd number of them (i.e.~$p_{n}$). In summary, we conclude $\vert S^{(\sigma_{n},\,\sigma_{n})}_{q_{n}} \vert \ge 2$, as claimed.
\end{proof}

\section{Constructing an unbounded subsequence of the counter trajectory, for an irrational rotation} \label{sec:main}

We are now in a position to formulate and prove our main result:
\begin{theorem} \label{thm_main}
Let $\sigma=[0; a_1, a_2, \dots] \in (0,1)$ be irrational such that its continued fraction expansion has the following properties: there exists a subsequence $(n_m)_{m \in \mathbb{N}}$ of $2\mathbb{N}$ such that 
\begin{itemize}
\item[(a)] $(a_{n_m})_{m \in \mathbb{N}}$ is unbounded
\item[(b)] for all $m \in \mathbb{N}$, both $p_{n_m}$ and $q_{n_m}$ are odd ( and thus the conditions of the Lemma ~\ref{lem:Sge2} are satisfied, for $n=n_{m}$);
\end{itemize}

Then there exists a full measure set of initial conditions $\Omega(\sigma) \subseteq [0,1)$ such that for each $x_0 \in \Omega(\sigma)$, one has
\begin{equation} \label{eq_mainthm_unbbcounter}
\limsup_{j \to \infty} \left|S_{j}\right| = \infty ~\mbox{.}
\end{equation}
\end{theorem}

\begin{remark} \label{remark_fullmeasinit}
Our proof of Theorem \ref{thm_main} implies the following explicit description of the set $\Omega(\sigma)$:
\begin{equation} \label{eq_initialcond_set}
\Omega(\sigma) := \bigcup_{Q \in \mathbb{N}} \left\{ x_0 \in [0,1) ~:~ \{x_0 q_{n_m} \} < 1 - \frac{1}{Q} ~\mbox{, for infinitely many $m \in \mathbb{N}$ } \right\} ~\mbox{.}
\end{equation}
Here $\{x\}: = x - \lfloor x \rfloor$ denotes the fractional part of $x \in [0,1)$. 

Observe that $\Omega(\sigma)$ is a set of {\em{full}} Lebesgue measure in $[0,1)$. Indeed, considering its complement
\begin{align}
[0,1) \setminus \Omega(\sigma) & = \{ x_0 \in [0,1) ~:~  \{x_0 q_{n_m} \} \to 1\} \\
   & \subseteq \left\{ x_0 \in [0,1) ~:~ \vert\vert\vert x_0 q_{n_m} \vert\vert\vert \to 0 \right\} \label{eq_fullmeasurecond} ~\mbox{,}
\end{align}
where $\vert\vert\vert x \vert\vert\vert:= \inf_{n \in \mathbb{Z}} \vert x -n \vert$ is the usual norm in $\mathbb{R} / \mathbb{Z}$, shows that the set on the right-hand side of (\ref{eq_fullmeasurecond}) is a proper subgroup of $\mathbb{R} / \mathbb{Z}$. Thus, a well known fact from harmonic analysis (see e.g. problem~14 in Sec.~1 of Katznelson's book \cite{Katznelson_book2004}) implies that (\ref{eq_fullmeasurecond}), and hence also $[0,1) \setminus \Omega(\sigma)$, has zero Lebesgue measure.
\end{remark}

Before turning to the proof of Theorem \ref{thm_main}, we comment on the existence of irrationals $\sigma$ described in that theorem. To construct explicit examples of such $\sigma$, we first recall that the continued fraction expansion of $\sigma=[0; a_1, a_2, \dots] \in (0,1)$ satisfies the recursion relations (see e.g. \cite{Khinchin_book_1997}):
\begin{align} 
p_{n} & = a_{n} p_{n-1} + p_{n-2} ~\mbox{,} \label{eq_recursion_p} \\
q_{n} & = a_{n} q_{n-1} + q_{n-2} ~\mbox{, for all $n \in \mathbb{N}$ ,} \label{eq_recursion_q}
\end{align}
with initial conditions
\begin{align} \label{eq_recursion_IC}
p_0 = 0 ~\mbox{, } p_{-1} = 1 ~\mbox{,} \nonumber \\
q_0 = 1 ~\mbox{, } p_{-1} = 0 ~\mbox{.}
\end{align}

Suppose that $(a_n)$ is a sequence in $(2 \mathbb{N} -1)$. Using induction, the recursion relations (\ref{eq_recursion_p})-(\ref{eq_recursion_q}) together with the initial conditions (\ref{eq_recursion_IC}) imply that
\begin{equation}
p_n = \begin{cases} \text{odd}\;, & ~\text{if $n \equiv 1,2 \mod 3$\,,} \\  \text{even} & ~\text{if $n \equiv \;0 \mod 3$\,,}  \end{cases}
\end{equation}
and
\begin{equation}
q_n = \begin{cases} \text{odd}\;, & ~\text{if $n \equiv 0,1 \mod 3$\,,} \\  \text{even}\;, & ~\text{if $n \equiv \;2 \mod 3$\,.}  \end{cases}
\end{equation}
The conditions of Theorem \ref{thm_main} are thus satisfied for $(n \equiv 1 \mod 3)$, specifically by letting
\begin{equation}
n_m = 3 (2 m -1) + 1 ~\mbox{, $m \in \mathbb{N}$ .} 
\end{equation}
In summary, we have shown that all irrationals $\sigma=[0; a_1, a_2, \dots] \in (0,1)$ for which the sequence of elements $(a_n)$ has {\em{odd parity}} satisfy the hypotheses of Theorem \ref{thm_main}. 

\begin{proof}[Proof of Theorem \ref{thm_main}]
Let $\sigma \in (0,1)$. Observe that since both conditions (a) and (b) are assumed to only hold along some {\em{sub}}sequence $(n_m)_{m \in \mathbb{N}}$ of $2\mathbb{N}$, possibly passing to a sub-subsequence one may replace hypothesis (a) with
\begin{equation}
\lim_{m \to \infty} a_{n_m} = \infty ~\mbox{.}
\end{equation}

Fix an initial condition $x_0 \in \Omega(\sigma)$, where $\Omega(\sigma)$ is described in (\ref{eq_initialcond_set}); in particular, there exists a $Q \in \mathbb{N}$ such that 
\begin{equation}
\{x_0 q_{n_m} \} < 1 - \frac{1}{Q} ~\mbox{, for infinitely many $m \in \mathbb{N}$ .} 
\end{equation}
Again, possibly passing to an appropriate subsequence, we may simply assume that 
\begin{equation}
\{x_0 q_{n_m} \} < 1 - \frac{1}{Q} ~\mbox{, for all $m \in \mathbb{N}$ .} 
\end{equation}

We recall two useful properties of continued fractions (see e.g. \cite{Khinchin_book_1997}): For all $n \in \mathbb{N}$, one has
\begin{equation} \label{eq_rateof conv_contfract}
\vert \sigma_n - \sigma \vert < \dfrac{1}{q_n q_{n+1}} ~\mbox{;}
\end{equation}
morevoer, since by hypothesis $n_m \in 2 \mathbb{N}$, one has
\begin{equation} \label{eq_evenconv}
\sigma_{n_m} < \sigma ~\mbox{, for all $m \in \mathbb{N}$ .}
\end{equation}

Now fix $m \in \mathbb{N}$, and let
\begin{align} \label{eq_subsequence_Q}
j^{(Q)}_{m} := \lfloor \frac{a_{n_{m}+1}}{Q} \rfloor q_{n_{m}} ~\mbox{.}
\end{align}

Then, for $1 \leq j < j^{(Q)}_{m}$, using (\ref{eq_rateof conv_contfract}), (\ref{eq_subsequence_Q}), and the recursion relation in (\ref{eq_recursion_q}),  we make the estimate
\begin{align} \label{eq_distancetraj}
\vert x_j^{(\sigma)} - x_j^{(\sigma_{n_m})} \vert < j^{(Q)}_{m} \vert \sigma - \sigma_m \vert \leq \dfrac{1}{Q q_{n_m}} ~\mbox{.}
\end{align}

Moreover, notice that the representation of the elements of the rational auxiliary trajectory in (\ref{k_0n})-(\ref{eq_represent_auxiltraj}) implies
\begin{align}
x_j^{(\sigma_{n_m})} & = \left( x_0 \mod \frac{1}{q_{n_m}} \right) +  \frac{k_{j,n_m}}{q_{n_m}} \nonumber \\
    & = \dfrac{ \{ x_0 q_{n_m} \}   }{q_{n_m}} +  \frac{k_{j,n_m}}{q_{n_m}} < \dfrac{ 1 - \frac{1}{Q} }{q_{n_m}} + \frac{k_{j,n_m}}{q_{n_m}} \label{eq_check_1}  \\
    & \leq \dfrac{ 1 - \frac{1}{Q} }{q_{n_m}} + \frac{q_{n_m} - 1}{q_{n_m}} \leq 1 - \dfrac{ 1 }{Q q_{n_m}} ~\mbox{, for } 0 \leq j \leq q_{n_m} - 1 \label{eq_auxiltrajcontrol} ~\mbox{.}
\end{align}

Thus, for $0 \leq j < j^{(Q)}_{m}$, we conclude from (\ref{eq_distancetraj}) and (\ref{eq_auxiltrajcontrol}) that
\begin{align*}
x^{(\sigma)}_{j} 
&
= \left( \underbrace{ j(\sigma - \sigma_{n_{m}}) }_{ < \frac{1}{Q q_{n_{m}} } } 
\quad+\quad  \underbrace{x^{(\sigma_{n_{m}})}_{j}}_{\le 1 - \frac{1}{Q q_{n_{m}}}  } \right) \mod 1
\nonumber
\\ & = j(\sigma - \sigma_{n_m}) +  x_j^{(\sigma_{n_m})}  \nonumber ~\mbox{,}
\end{align*}
whence
\begin{align}
x^{(\sigma)}_{j}  & < x^{(\sigma_{n_{m}})}_{j} + \frac{1}{Q q_{n_{m}}} ~\mbox{, } 0 \leq j < j^{(Q)}_{m} ~\mbox{, } \label{main_inequalities_1}\\
x^{(\sigma)}_{j}  & \ge x^{(\sigma_{n_{m}})}_{j} + (\sigma-\sigma_{n_{m}})  ~\mbox{, } 1 \le j < j^{(Q)}_{m} ~\mbox{.} \label{main_inequalities_2}
\end{align}

Let us now prove that 
\begin{align} \label{f_Eq_f_n}
f^{(\sigma,\,\sigma)}_{j}   = f^{(\sigma_{n_{m}},\,\sigma_{n_{m}})}_{j} ~\mbox{, for  } 0 \leq j <  j_{m}^{(Q)} ~\mbox{.}
\end{align}

Using (\ref{main_inequalities_1}), one has for $j < j_m^{(Q)}$
\begin{align}
\begin{array}{rcl}
f^{(\sigma_{n_{m}},\,\sigma_{n_{m}})}_{j} =  +1  
&
\Leftrightarrow 
&
x^{(\sigma_{n_{m}})}_{j} < \sigma_{n_{m}} ~\mbox{ and } k_{j,n_m} \leq p_{n_m} - 1
\\
&
\Rightarrow
&
x^{(\sigma_{n_{m}})}_{j} +  \frac{1}{Q q_{n_{m}}} \stackrel{(\ref{eq_check_1} )}{\leq} \sigma_{n_{m}} 
\\
&
\stackrel{\text{\eqref{main_inequalities_1}}}{\Rightarrow} 
&
x^{(\sigma)}_{j} \leq \sigma_{n_{m}} \stackrel{ (\ref{eq_evenconv})   }{<} \sigma
\\
&
\Rightarrow 
&
f^{(\sigma,\,\sigma)}_{j} = +1\,.
\end{array}
\,\,
\label{fP1}
\end{align}
Similarly, for $1 \le j<j_{m}^{(Q)}$, (\ref{main_inequalities_2}) yields that
\begin{align}
\begin{array}{rcl}
f^{(\sigma_{n_{m}},\,\sigma_{n_{m}})}_{j} =  -1  
&
\Leftrightarrow
& x^{(\sigma_{n_{m}})}_{j} \ge \sigma_{n_{m}} 
\\
&
\Leftrightarrow
& x^{(\sigma_{n_{m}})}_{j} + (\sigma-\sigma_{n_{m}}) \ge \sigma 
\\
&
\stackrel{\text{\eqref{main_inequalities_2}}}{\Rightarrow} 
&
x^{(\sigma)}_{j} \ge \sigma 
\\
&
\Leftrightarrow
& f^{(\sigma,\,\sigma)}_{j} = -1 \,.
\end{array}
\label{fM1}
\end{align}
Finally, since $x^{(\sigma_{n_{m}})}_{0} = x^{(\sigma)}_{0} = x_{0}$,
\begin{align}
f^{(\sigma_{n_{m}},\,\sigma_{n_{m}})}_{j} \stackrel{j=0}{=} f^{(\sigma,\,\sigma)}_{j} 
\,\,.
\label{f_j0}
\end{align}
In summary, \eqref{fP1}, \eqref{fM1}, and \eqref{f_j0} validate \eqref{f_Eq_f_n}. 

Therefore, for $j < j_{m}^{(Q)}$, any conclusion about the auxiliary sequence $f^{(\sigma_{n_{m}},\,\sigma_{n_{m}})}_{j}$ will be valid for the sequence $f^{(\sigma,\,\sigma)}_{j}$. By the definition of the increments $\epsilon^{(\alpha,\,\beta)}_{j}$ in (\ref{epsilon}), (\ref{f_Eq_f_n}) implies 
\begin{align} \label{f_Eq_epsilon_n}
\epsilon^{(\sigma,\,\sigma)}_{j}   = \epsilon^{(\sigma_{n_{m}},\,\sigma_{n_{m}})}_{j} ~\mbox{, for  } 0 \leq j <  j_{m}^{(Q)} ~\mbox{.}
\end{align}
Finally, since the counter $S^{(\alpha,\,\beta)}_{j}$ is only sensitive to the values of $\epsilon^{(\alpha,\,\beta)}_{j'}$ for $j' :  0 \le j' \le j-1$ (see \eqref{S}), we obtain
\begin{align}
S^{(\sigma,\,\sigma)}_{j} =  S^{(\sigma_{n_{m}},\,\sigma_{n_{m}})}_{j}  \text{  , for  }  j \le j_{m}^{(Q)} 
\,\,.
\label{main_equality}
\end{align}

In particular, the property \eqref{counter_growth_01} will be valid for $S^{(\sigma,\,\sigma)}_{N q_{n_{m}}}$, with $N = \lfloor a_{n_{m}+1} / Q \rfloor$. 
To see this, we combine \eqref{main_equality}, condition (b) of Theorem \ref{thm_main}, and 
Proposition \ref{lem:Sge2} to get 
\begin{align}
\left|S^{(\sigma,\,\sigma)}_{j_{m}^{(Q)}}\right| 
= \left|S^{(\sigma_{n_{m}},\,\sigma_{n_{m}})}_{j_{m}^{(Q)}} \right|
=  \left|S^{(\sigma_{n_{m}},\,\sigma_{n_{m}})}_{\lfloor \frac{a_{n_{m}+1}}{Q} \rfloor q_{n_{m}}}\right|  = \lfloor \frac{a_{n_{m}+1}}{Q} \rfloor \left|S^{(\sigma_{n_{m}},\,\sigma_{n_{m}})}_{q_{n_{m}}}\right|
\geq 2 \lfloor \frac{a_{n_{m}+1}}{Q} \rfloor
\,\,.
\label{growth_Bis}
\end{align}
By hypothesis (a) of Theorem \ref{thm_main}, this implies the claim in (\ref{eq_mainthm_unbbcounter}). 
\end{proof}

What follows from the Theorem \ref{thm_main}, is that overall, the sequence $S^{(\sigma,\,\sigma)}_{j}$ is unbounded, in a seeming 
contradiction to the observation \cite{wang2014_042918}.

\begin{remark} Inequality (\ref{growth_Bis}) implies, using (\ref{eq_subsequence_Q})
  and the properties of continued fractions \cite{Khinchin_book_1997}, that
  the growth along the subsequence $j_{m}^{(Q)}$ is at least
  logarithmic. Note that the subsequence itself is at least
  exponentially sparse. However, choosing a very fast growing sequence
  $a_{n_{m}+1}$ one can also ensure any faster sublinear growth rate
  along $j_{m}^{(Q)},$ although it will come at the expense of making
  the subsequence of growth even sparser.

\end{remark}

We mention that it remains an open question whether the condition on the rotations $\sigma$ in Theorem \ref{thm_main}, or the respective condition on the initial conditions in (\ref{eq_initialcond_set}), are indeed necessary. In particular, an interesting question for future research may be to examine the case of $\sigma$ of bounded type (e.g. take $\sigma$ equal to the golden mean), which however requires development of a different proof strategy.

\section{Conclusions: implications of our results on irrational rotations for the Wang-Casati-Prosen counter in a $\ang{45} \!\! : \!\ang{45} \!\! : \! \ang{90}$ billiard}
\label{sec:conclusions}
The result of Theorem \ref{thm_main_main} suggests that the counter $K_{i}$ (see \eqref{K}) of \cite{wang2014_042918} shows absence of localization of the counter,  
for a $45^{\circ}\!\!:\!45^{\circ}\!\!:\!90^{\circ}$ billiards, at least for one particular initial condition, and, if proven, for all rational initial conditions. 
It seems unlikely localization can reemerge in the generic $\tilde{\alpha}/2\!:\!(90^{\circ}-\tilde{\alpha}/2)\!:\!90^{\circ}$ case of \cite{wang2014_042918}.  Numerical results of \cite{wang2014_042918} 
do however suggest localization. A potential resolution for this contradiction can be offered by a probable very slow growth of the counter. In what follows, we will 
provide an explicit example that confirms this slowness. 

\begin{figure}[h]
\begin{center}
\includegraphics[width=0.8\textwidth]{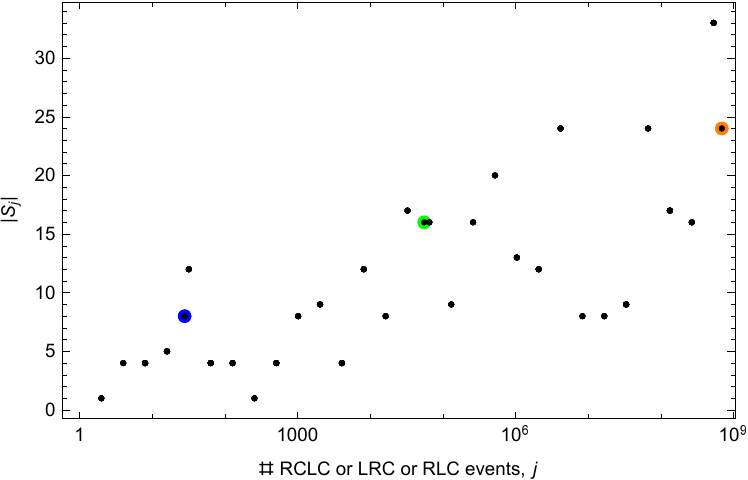}
\end{center}
\caption
{
Counter $S_{j}$ as a function of ``time'' $j$. We show 200 randomly chosen instances. In addition, 
we show $S_{j}$ at $j_{m=0,\,1,\,2}$ (see \eqref{eq_subsequence_Q} for $Q=1$ and $x_0 = 0$), along with the lower bound 
\eqref{growth_Bis}. Blue, green, and orange dots correspond to $m=0,\,1,\,2$, respectively.  
}
\label{f:S}
\end{figure}

Consider
\begin{align}
\sigma 
&=  [0;\,2,\,3,\,\ldots,\,n+1,\,\ldots]
\label{sigma_fixed}
\\
&\approx 0.433127
\nonumber
\,\,.
\end{align}
Using the recurrence relations (\ref{eq_recursion_p})-(\ref{eq_recursion_q}) with the initial conditions (\ref{eq_recursion_IC}) one can easily show that the subsequence of the rational approximants with
\begin{align*}
n_{m} = 4 m + 2
\end{align*}
will satisfy all conditions of Theorem~\ref{thm_main}. Thus, for $x_0 = 0$, the proof of Theorem~\ref{thm_main} (use (\ref{eq_subsequence_Q}) with $Q=1$) shows that $S_{j_{m}}$ is unbounded for
\begin{align*}
j_{m} = a_{4m+3} \, q_{4m+2}
\,\,
\end{align*}

The first three temporal instances showing a provable growth are
\begin{align*}
\begin{array}{ccccc}
m= & 0 & 1 & 2 & \ldots
\\
j= & 28 & 55688 & 695991252 & \ldots
\\
S_{j} = & +8 &  +16 &  +24 & \ldots
\\
\left|S_{j}\right| \ge 2 \, a_{4 m + 3} = & 8 & 16 & 24 & \ldots 
\end{array}
\,\,
\end{align*}
(see Fig.~\ref{f:S}).
The last line gives the lower bound \eqref{growth_Bis}. Notice that $\left|S_{j}\right|$ stays at its lowest value allowed, something 
that we can neither prove nor disprove at the moment. 
%
%
%
\begin{figure}[h]
\begin{center}
\includegraphics[width=0.8\textwidth]{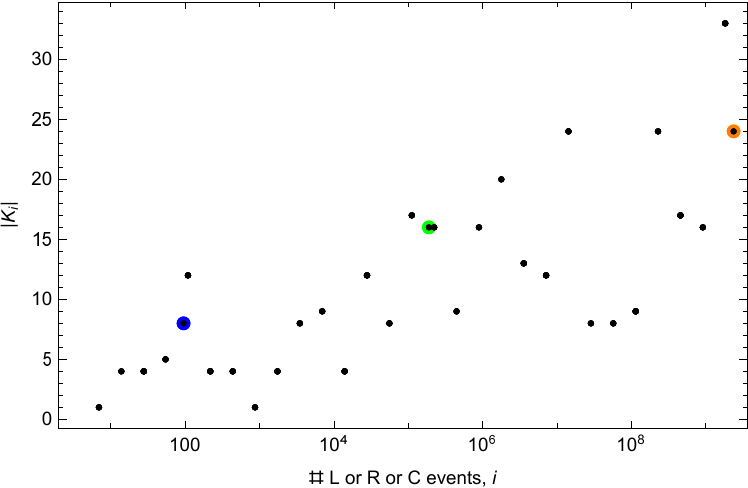}
\end{center}
\caption
{
Counter $K_{i}$ of  \cite{wang2014_042918} as a function of ``time'' $i$. In terms of \cite{wang2014_042918}, 
the rest is the same as in
Fig.~\ref{f:S}
}
\label{f:K}
\end{figure}

The growth of the counter, while present, appears to be slow. We suggest the following order of magnitude 
estimate for the growth of the counter. From (\ref{eq_recursion_q}), we get 
\begin{align*}
\ln(q_{n}) \sim \ln(a_{n}!) \approx a_{n} \ln(a_{n})  \Rightarrow \ln(a_{n}) \sim \frac{\ln(q_{n})}{\ln(\ln(q_{n}))} \Rightarrow S_{j} 
\gtrsim 2 \frac{\ln(j)}{\ln(\ln(j))-1}  
\,\,.
\end{align*}

The corresponding values of the ``physical'' counter in question, from Ref.~\cite{wang2014_042918}, can also be computed (see Fig.~\ref{f:K}).:
\begin{align*}
\begin{array}{ccccc}
m= & 0 & 1 & 2 & \ldots
\\
i= &  96 &  191184  &  2389426656  & \ldots
\\
K_{i} = & -8 &  -16 &  -24 & \ldots
\\
\left|K_{i}\right| \ge 2 \, a_{4 m + 3} = & 8 & 16 & 24 & \ldots 
\end{array}
\,\,.
\end{align*}
Recall that $S_{j} = (-1)^{i(j)+1} K_{i(j)}$. Also, $i(j)$ can be estimated as
\begin{align*}
i(j) \approx (\sigma \times 4 + (1-\sigma) \times 3)\times j \approx 3.43313\times j 
\,\,,
\end{align*}
producing $j = 96.1, \,191184.0,\,  2389426656.0  ,\,\ldots$ in the second line above. The estimate is using (a) the ergodicity 
of the irrational rotations, leading to $\text{Prob}\left(x \in \mathcal{I}_{\text{I}} \equiv [0,\,\sigma[\, \right)= \sigma$ and 
$\text{Prob}\left(x \in \mathcal{I}_{\text{II}} \equiv [\sigma,\,1[\, \right) = 1 - \sigma$ and (b) the fact that $x \in \mathcal{I}_{\text{I}}$
corresponds to a four-physical-events RCLC fragment, while $x \in \mathcal{I}_{\text{II}}$ corresponds to either LRC or RLC 
three-physical-events fragments.


\section*{Acknowledgements}

We thank Toma\v{z}  Prosen \maxim{}{and Giulio Casati} for
valuable discussions. S.J. was a 2020-21 Simons fellow. Her work was also partially supported by NSF DMS-2052899, DMS-2155211, and Simons 681675.
LPL is a member of DIM
SIRTEQ (Science et Ing\'{e}nierie en R\'{e}gion \^{I}le-de-France
pour les Technologies Quantiques).
M.O. was supported by the NSF under grants PHY-1607221 and PHY-1912542.

\bibliography{zai_counter_01.bib}

\end{document}